\documentclass[sigconf]{acmart}  	
\renewcommand\footnotetextcopyrightpermission[1]{} 
\usepackage{geometry}                		
\usepackage{graphicx}				
\usepackage{amssymb}
\usepackage{amsmath}
\usepackage{amsthm}
\usepackage{bm}
\usepackage{stackrel}
\usepackage[normalsize]{subfigure}

\usepackage{thmtools}
\usepackage{thm-restate}
\usepackage{hyperref}
\usepackage{enumitem}

\author{Benjamin Berg}
\affiliation{%
\institution{Carnegie Mellon University}
}
\email{bsberg@cs.cmu.edu}
\author{Rein Vesilo}
\affiliation{%
\institution{Macquarie University}
}
\email{rein.vesilo@mq.edu.au}
\author{Mor Harchol-Balter}
\affiliation{%
\institution{Carnegie Mellon University}
}
\email{harchol@cs.cmu.edu}

\title{heSRPT: Optimal Parallel Scheduling of Jobs With Known Sizes}
\date{}
\begin{document}
\begin{abstract}
When parallelizing a set of jobs across many servers, one must balance a trade-off between granting priority to short jobs and maintaining the overall efficiency of the system.
When the goal is to minimize the mean flow time of a set of jobs, it is usually the case that one wants to complete short jobs before long jobs.
However, since jobs usually cannot be parallelized with perfect efficiency, granting strict priority to the short jobs can result in very low system efficiency which in turn hurts the mean flow time across jobs.

In this paper, we derive the optimal policy for allocating servers to jobs at every moment in time in order to minimize mean flow time across jobs.
We assume that jobs follow a sublinear, concave speedup function, and hence jobs experience diminishing returns from being allocated additional servers.
We show that the optimal policy, heSRPT, will complete jobs according to their size order, but maintains overall system efficiency by allocating some servers to each job at every moment in time.
We compare heSRPT with state-of-the-art allocation policies from the literature and show that heSRPT outperforms its competitors by at least 30\%, and often by much more.
\end{abstract}

\maketitle
\clearpage
\section{Introduction}
Nearly all modern data centers serve workloads which are capable of exploiting parallelism.  
When a job parallelizes across multiple servers it will complete more quickly.
However, it is unclear how best to share a limited number of servers between many parallelizable jobs.

In this paper we consider a typical scenario where a data center composed of $N$ servers will be tasked with completing a set of $M$ parallelizable jobs, where typically $M$ is much smaller than $N$.
In our scenario, each job has a different inherent size (service requirement) which is known up front to the system.
In addition, each job can be run on any number of servers at any moment in time.
These assumptions are reasonable for many parallelizable workloads such as training neural networks using TensorFlow \cite{abadi2016tensorflow,lin2018model}. 
Our goal in this paper is to allocate servers to jobs so as to minimize the \emph{mean flow time} across all jobs, where the flow time of a job is the time until the job leaves the system.\footnote{We will also consider the problem of minimizing \emph{makespan}, but this turns out to be fairly easy and thus we defer that discussion.}
What makes this problem difficult is that jobs receive a concave, sublinear speedup from parallelization -- jobs have a decreasing marginal benefit from being allocated additional servers (see Figure \ref{fig:speedups}).
Hence, in choosing a job to receive each additional server, one must keep the overall efficiency of the system in mind.
The goal of this paper is to determine the optimal allocation of servers to jobs where all jobs follow a realistic sublinear speedup function.

It is clear that the optimal allocation policy will depend heavily on the jobs' speedup -- how parallelizable the jobs being run are.  
To see this, first consider the case where jobs are \emph{embarrassingly parallel}.  
In this case, we observe that the entire data center can be viewed as a \emph{single server} that can be perfectly utilized by or shared between jobs.
Hence, from the single server scheduling literature, it is known that the Shortest Remaining Processing Time policy (SRPT) will minimize the mean flow time across jobs \cite{smith1978new}.
By contrast, if we consider the case where jobs are hardly parallelizable, a single job receives very little benefit from additional servers.
In this case, the optimal policy is to divide the system equally between jobs, a policy called EQUI.  
In practice, a realistic speedup function usually lies somewhere between these two extremes and thus we must balance a trade-off between the SRPT and EQUI policies in order to minimize mean flow time.  
Specifically, since jobs are \emph{partially} parallelizable, it is still beneficial to allocate more servers to smaller jobs than to large jobs.  
The optimal policy with respect to mean flow time must split the difference between these policies, figuring out how to favor short jobs while still respecting the overall efficiency of the system.

\begin{figure}[ht]
\centering
\includegraphics[width=.45\textwidth]{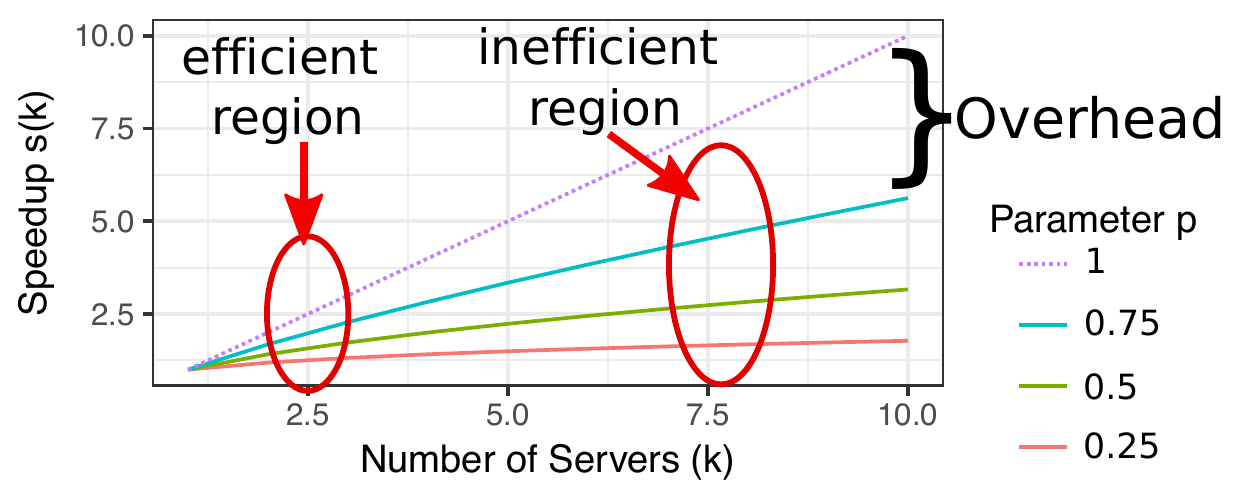}
    \caption{A variety of speedup functions of the form $s(k)=k^p$, shown with varying values of $p$.
    When $p=1$ we say that jobs are \emph{embarrassingly parallel}, and hence we consider cases where $0 < p <1$.
    Note that all functions in this family are concave and lie below the embarrassingly parallel speedup function ($p=1$).}

\label{fig:speedups}
\end{figure}

\subsection*{Prior Work}
Despite the prevalence of parallelizable data center workloads, it is not known, in general, how to optimally allocate servers across a set of parallelizable jobs.  
The state-of-the-art in production systems is to let the user decide their job's allocation by reserving the resources they desire \cite{verma2015large}, and then allowing the system to pack jobs onto servers \cite{ren2016clairvoyant}.
This allows users to reserve resources greedily, and can lead to low system efficiency.
We seek to improve upon the status quo by allowing the system to choose the allocation of servers to each job.

The closest work to the results presented in this paper is \cite{lin2018model}, which considers jobs which follow a realistic speedup function and have known, generally distributed sizes.
Similarly to our work, \cite{lin2018model} allows server allocations to change over time.
While \cite{lin2018model} proposes and evaluates some heuristic policies, they make no theoretical guarantee about the performance of their policies.

Other related work from the performance modeling community, \cite{berg2018}, assumes that jobs follow a concave speedup function and allows server allocations to change over time.
However, unlike our work, \cite{berg2018} assumes that job sizes are \emph{unknown} and are drawn from an exponential distribution.
\cite{berg2018} concludes that EQUI is the optimal allocation policy.
However, assuming unknown exponentially distributed job sizes is highly pessimistic since this means job sizes are impossible to predict, even as a job ages.

There has also been work on how to allocate servers to jobs which follow arbitrary speedup functions.
Of this work, the closest to our model is \cite{im2016competitively} which also considers jobs of known size. Inversely, \cite{Edmonds1999SchedulingIT, edmonds2009scalably,Agrawal:2016:SPJ:2935764.2935782} all consider jobs of unknown size.
This work is all through the lens of competitive analysis, which assumes that job sizes, arrival times, and even speedup functions are adversarially chosen.  
This work concludes that a variant of EQUI is $(1+\epsilon)$-speed $o(1)$-competitive with the optimal policy when job sizes are unknown \cite{edmonds2009scalably}, and that a combination of SRPT and EQUI is $o(1)$-competitive when job sizes are known \cite{im2016competitively}.

The SPAA community often models each job as a DAG of interdependent tasks \cite{blumofe1999scheduling,bampis2014note,bodik2014brief,narlikar1999space}.
This DAG encodes precedence constraints between tasks, and thus implies how parallelizable a job is at every moment in time.
Given such a detailed model, it is not even clear how to optimally schedule a \emph{single} DAG job on many servers in order to minimize the job's completion time \cite{chowdhury2013oblivious}.
The problem only gets harder if tasks are allowed to run on multiple servers \cite{du1989complexity,chen2018improved}.
Other prior work considers the scheduling of several DAG jobs onto several servers in order to minimize the mean flow time, makespan, or profit of the set of jobs \cite{chudak1999approximation,hall1997scheduling,chekuri2001approximation,agrawal2018scheduling}.
All of this work is fundamentally different from our model in that it models parallelism in a much more fine-grained way.
Our hope is that by modeling parallelism through the use of speedup functions, we can address problems that would be intractable in the DAG model.

Our model also shares some similarities with the coflow scheduling problem \cite{jahanjou2017asymptotically, chowdhury2014efficient, qiu2015minimizing,shafiee2017brief,khuller2016brief}.  In coflow scheduling, one must allocate a continuously divisible resource, link bandwidth, to a set of network flows to minimize mean flow time.   Unlike our model, there is no explicit notion of a flow's speedup function here.  Given that this problem is NP-Hard, prior work examines the problem via heuristic policies \cite{chowdhury2014efficient}, and approximation algorithms \cite{qiu2015minimizing,jahanjou2017asymptotically}.

\begin{figure}[t]
\includegraphics[width=.47\textwidth]{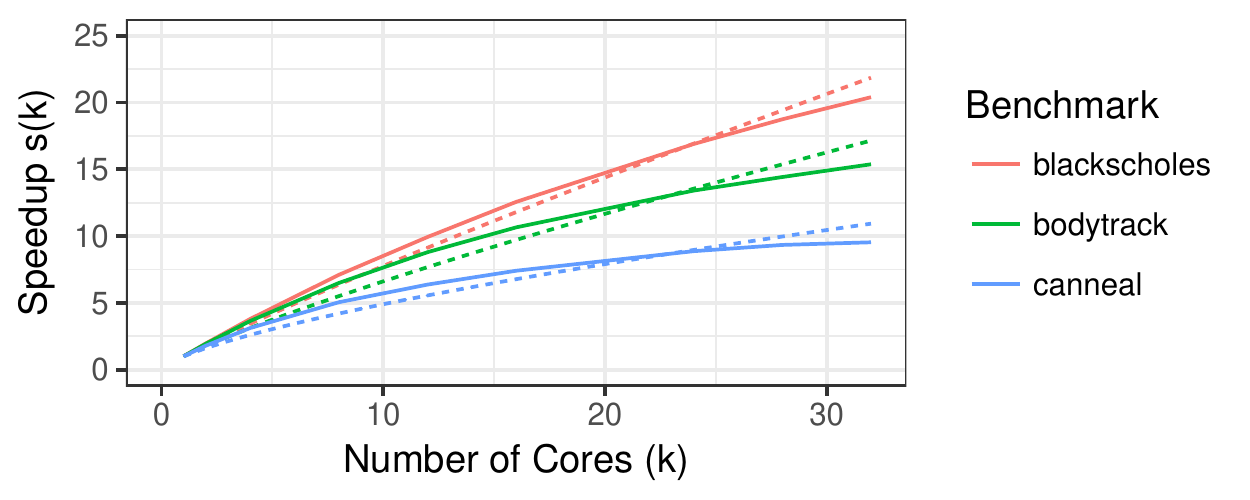}
    \caption{Various speedup functions of the form $s(k)=k^p$ (dotted lines) which have been fit to real speedup curves (solid lines) measured from jobs in the PARSEC-3 parallel benchmarks\cite{zhan2017parsec3}.  The three jobs, blackscholes, bodytrack, and canneal, are best fit by the functions where $p=.89$, $p=.82$, and $p=.69$ respectively.}

\label{fig:parsec}
\end{figure}

\begin{figure*}[ht]
\centering
\includegraphics[width=.75\textwidth]{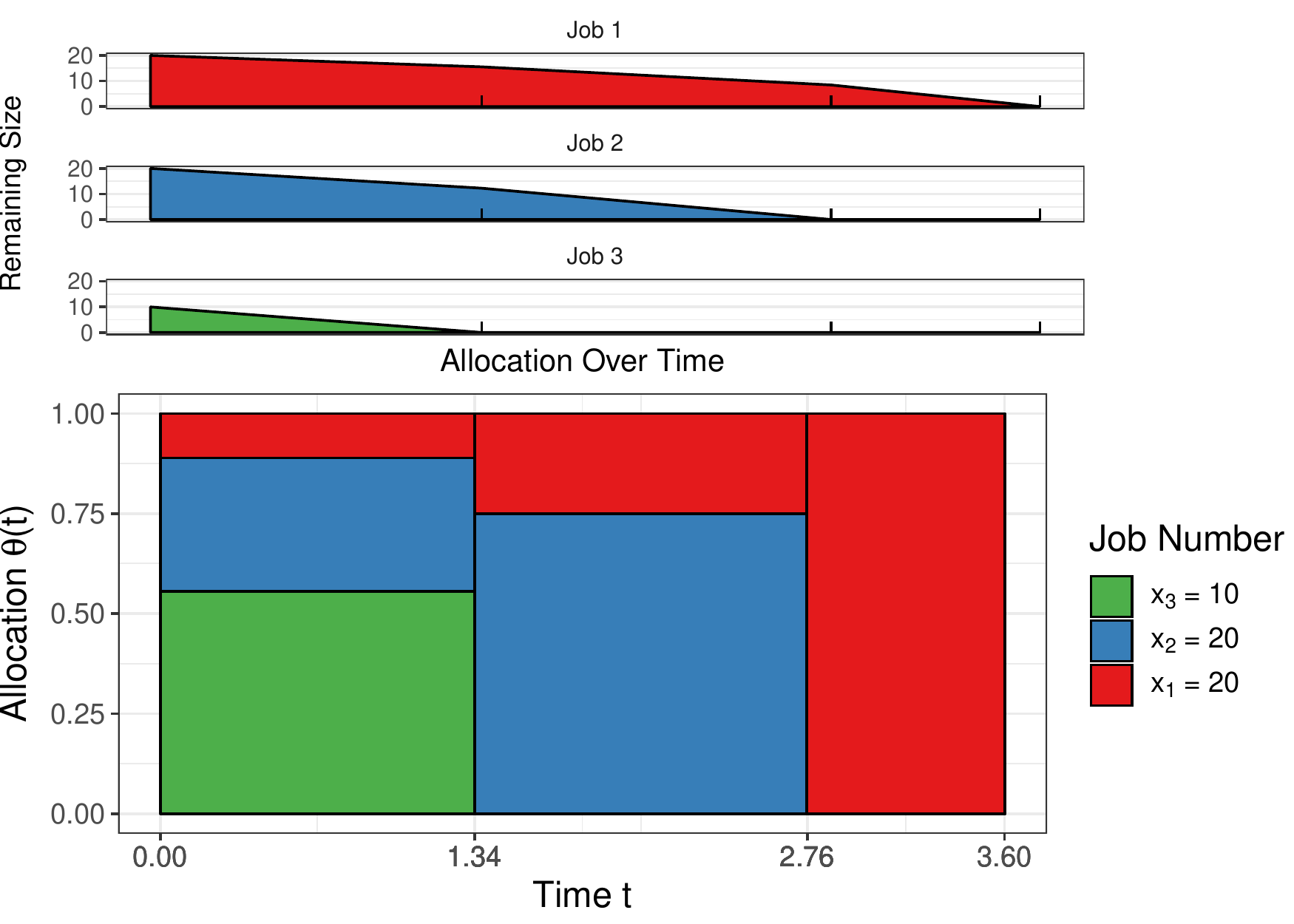}
\label{fig:3jobs}
    \caption{An example of the optimal allocation policy, heSRPT, completing a set of three jobs where the speedup function is $s(k)=k^{.5}$ and $N=500$.
    The top graph shows the remaining size of each job over time, while the bottom graph shows the allocation, $\bm{\theta}(t)$, over time.
    Jobs are finished in shortest job first order, but rather than allocating the whole system to the shortest job, heSRPT optimally shares the system between all active jobs.}
\end{figure*}
\subsection*{Our Model}
Our model assumes that there are $N$ identical servers which must be allocated across a set of $M$ parallelizable jobs.  All $M$ jobs are present at time $t=0$.
Job $i$ is assumed to have some inherent size $x_i$ where, WLOG,
$$x_1 \geq x_2 \geq \ldots \geq x_M.$$

We assume that all jobs follow the \emph{same} speedup function, $s : \mathbb{R}^+ \rightarrow \mathbb{R}^+$, which is of the form
$$s(k) = k^p \qquad$$
for some $0 <p<1$.
Specifically, if a job $i$ of size $x_i$ is allocated $k$ servers, it will complete at time
$$\frac{x_i}{s(k)}.$$
In general, the number of servers allocated to a job can change over the course of the job's lifetime.
It therefore helps to think of $s(k)$ as a \emph{rate}\footnote{WLOG we assume the service rate of a single server to be 1.  More generally, we could assume the rate of each server to be $\mu$, which would simply replace $s(k)$ by $s(k)\mu$ in every formula} of service where the remaining size of job $i$ after running on $k$ servers for a length of time $t$ is
$$x_i - t \cdot s(k).$$
We choose the family of functions $s(k)=k^p$ because they are (i) sublinear and concave, (ii) can be fit to a variety of empirically measured speedup functions (see Figure \ref{fig:parsec}), and (iii) simplify the analysis.  Note that \cite{Hill:2008:ALM:1449375.1449387} assumes $s(k) = k^p$ where $p=0.5$ and explicitly notes that using speedup functions of another form does not significantly impact their results. 

In general, we assume that there is some policy, $P$, which allocates servers to jobs at every time, $t$.
When we talk about the system at time $t$, we will use $m^P(t)$ to denote the number of remaining jobs in the system, and $x^P_i(t)$ to denote the remaining size of job $i$.  We also denote the completion time of job $i$ under policy $P$ as $T^P_i$.  When the policy $P$ is implied, we will drop the superscript.

In general, a policy $P$ will complete jobs in a particular order.  We define
$$C = (c_1, c_2, \ldots, c_M)$$ 
to be a permutation of $1, 2, \ldots, M$ which specifies a \emph{completion order} of jobs.  If a policy follows the completion order, $C$, then for any $i < j$, job $c_i$ completes \emph{after} job $c_j$.
Specifically, job $c_1$ is the last job to complete and job $c_M$ is the first job to complete under the completion order $C$.

A policy $P$ is said to be \emph{optimal with respect to the completion order} $C$ if it achieves the lowest mean flow time of any policy which follows the completion order $C$.

We will assume that the number of servers allocated to a job need not be discrete.
In general, we will think of the $N$ servers as a \emph{single, continuously divisible resource}. 
Hence, the policy $P$ with completion order $C$ can be defined by an \emph{allocation function} $\bm{\theta}^P(t)$ where
$$\bm{\theta}^P(t) = (\theta^P_{c_1}(t), \theta^P_{c_2}(t), \ldots, \theta^P_{c_{m(t)}}(t)).$$
Here, $0 \leq \theta^P_i(t) \leq 1$ for each job $i$, and $\sum_{i=1}^{m(t)} \theta_i^P(t) \leq 1$.  An allocation of $\theta^P_{c_i}(t)$ denotes that under policy $P$, at time $t$, job $c_i$ receives a speedup of $s(\theta^P_{c_i}(t)\cdot N)$.

We will denote the allocation function of the optimal policy which minimizes mean flow time as $\bm{\theta}^*(t)$.  Similarly, we let $m^*(t)$, $x^*_i(t)$, $C^*$, and $T^*_i$ denote the corresponding quantities under the optimal policy.  

\subsection*{Why  Server Allocation is Counter-intuitive }
Consider a simple system with $N=10$ servers and $M=2$ identical jobs of size 1, where $s(k) = k^{.5}$, and where we wish to minimize mean flow time.
One intuitive argument would be that, since everything in this system is symmetric, the optimal allocation should be symmetric.
Hence, one might think to allocate half the servers to job one and half the servers to job two.
Interestingly, while this does minimize the \emph{makespan} of the jobs, it does \emph{not} minimize their flow time.
Alternately, a queueing theorist might look at the same problem and say that to minimize flow time, we should use the SRPT policy, allocating all servers to job one and then all servers to job two.
However, this causes the system to be very inefficient.
We will show that the optimal policy in this case is to allocate $75\%$ of the servers to job one and $25\%$ of the servers to job two.
In our simple, symmetric system, the optimal allocation is very asymmetric!
Note that this asymmetry is \emph{not} an artifact of the form of the speedup function used.
If we had instead assumed that $s$ was Amdahl's Law \cite{Hill:2008:ALM:1449375.1449387} with a parallelizable fraction of $f=.9$, the optimal split is to allocate $63.5\%$ of the system to one of the jobs.
If we imagine a set of $M$ arbitrarily sized jobs, one suspects that the optimal policy again favors shorter jobs, but calculating the exact allocations for this policy is not trivial.

\subsection*{Contributions}
The contributions of this paper are as follows:
\begin{itemize}[leftmargin=.35cm]
\item We derive the first closed-form expression for the optimal allocation of servers to jobs which minimizes mean flow time across jobs.  At any moment in time $t$ we define
         
$$\bm{\theta}^*(t) = (\theta^*_1(t), \theta^*_2(t), \ldots , \theta^*_m(t)),$$
        where $m=m^*(t)$ denotes the number of remaining jobs at time $t$ under the optimal allocation, and where $\theta^*_i(t)$ denotes the fraction of the $N$ servers allocated to job $i$ at time $t$.
Our optimal allocation balances the size-awareness of SRPT and the high efficiency of EQUI.  
We thus refer to our optimal policy as High Efficiency SRPT (heSRPT) (see Theorem \ref{thm:opt}).
        We also provide a closed-form expression for the mean flow time under heSRPT (see Theorem \ref{thm:optt}).
\item In the process of deriving $\bm{\theta}^*(t)$ we also prove many interesting properties of the optimal allocation function.
We prove the scale-free property in Section \ref{sec:sf}, derive the completion order of the optimal policy in Section \ref{sec:comp}, and prove the size-invariant property in Section \ref{sec:opt}.
\item While we can analyze the mean flow time under heSRPT, other policies in the literature such as HELL \cite{lin2018model}, and KNEE \cite{lin2018model} are not analytically tractable.
We therefore perform a numerical evaluation comparing heSRPT to the other proposed policies in the literature (see Section \ref{sec:eval}).
\item We have thus far focused on minimizing mean flow time, however one might also ask how to minimize the makespan of the set of jobs -- the completion time of the last job.  Minimizing makespan turns out to be easy in our setting.  We therefore present makespan minimization as a warm-up (see Section \ref{sec:makespan}).  We find that minimizing makespan favors \emph{long} jobs while maintaining high efficiency.  Thus, we name the optimal policy for minimizing makespan High Efficiency Longest Remaining Processing Time (heLRPT).
\end{itemize}

\section{A Warm-up: Minimizing Makespan}
\label{sec:makespan}
While we have thus far discussed the problem of minimizing the mean flow time of a set of $M$ jobs, this section will discuss the simpler problem of instead minimizing the \emph{makespan} of the jobs -- the time until the last job completes.
Makespan is important in applications such as MapReduce, where a scheduler tries to minimize the makespan of a set of parallelizable ``map tasks'' \cite{zhu2014minimizing} which must be all be completed before the results of the data analysis can be retrieved.
The problem of minimizing makespan, while hard in general, turns out to be fairly simple in our model.  We will prove an important property of the optimal policy in Theorem \ref{thm:equal}, and use this property in Theorem \ref{thm:helrpt} to derive the exact allocation function, $\bm{\gamma}^*(t)$, which minimizes makespan.

\subsection{Favoring Long Jobs Using heLRPT}
We begin by proving that the optimal policy with respect to makespan must complete all jobs at the same time.  This is stated formally in Theorem \ref{thm:equal}.

\begin{restatable}{thm}{equal}
    \label{thm:equal}
    Let $T^*_i$ be the completion time of job $i$ under the allocation function $\bm{\gamma}^*(t)$ which minimizes makespan.  Then,
    $$T^*_i = T^*_j \qquad \forall 0 < i,j \leq M.$$
\end{restatable}
\begin{proof}
    Assume for contradiction that not all jobs complete at the same time under $\bm{\gamma}^*(t)$.
    Let $i$ be the first job to complete and let $j$ be the last job to complete.
    We can now imagine an allocation function, $\bm{\gamma}'(t)$, which reallocates some fraction of the system from job $i$ to job $j$.
    We will choose some $\epsilon > 0$ and say that 
    $$
    \gamma'_i(t) = \gamma^*_i(t)-\epsilon \qquad \forall t \in [0,T^*_i] : \gamma^*_i(t) \geq \epsilon.
    $$
    $\bm{\gamma}'$ then divides the extra $\epsilon$ fraction of the system equally amongst all jobs that finished after $T^*_i$ under the optimal policy, reducing the completion time of all of these jobs while only hurting $i$ slightly.
    We can choose $\epsilon$ to be small enough that $\bm{\gamma}' \neq \bm{\gamma}^*$.
    Furthermore, since we can choose $\epsilon$ to arbitrarily small, the increase in $T'_i$ can be made arbitrarily small.
    Since $\bm{\gamma}'$ improves the makespan over $\bm{\gamma}^*$, we have a contradiction.
\end{proof}

Theorem \ref{thm:equal} tells us that, in finding the optimal policy to minimize makespan, we must only consider policies under which all job completion times are equal.
Since all job completion times are convex functions of their allocations, there exists only one policy which equalizes the completion time of the $M$ jobs.
This policy must therefore be optimal.
Theorem \ref{thm:helrpt} presents a closed-form for the optimal allocation function, $\bm{\gamma}^*(t)$.
\begin{restatable}{thm}{helrpt}
    \label{thm:helrpt}
    Let $\bm{X}=\{x_1, x_2 \ldots x_M\}$ and let 
    $$T^*_{max} = \max_{1\leq i \leq M} T^*_i$$
    be the optimal makespan.  Then
    $$T^*_{max} = ||\bm{X}||_{1/p}$$
    and the optimal policy with respect to makespan is given by the allocation function
    $$\gamma^*_i(t) = \frac{x_i^{1/p}}{\sum_{j=1}^M x_j^{1/p}} \qquad \forall t \leq T^*_{max}, 1 \leq i \leq M$$
\end{restatable}
\begin{proof}
    The proof is straightforward, see Appendix \ref{sec:moremake}. 
\end{proof}

Theorem \ref{thm:helrpt} says that the optimal policy will need to allocate a larger fraction of the system to longer jobs at every moment in time.  
This is similar to the Longest Remaining Processing Time (LRPT) policy, except that LRPT gives strict priority to the (one) largest job in the system.
We thus refer to the optimal policy for minimizing makespan as High Efficiency LRPT (heLRPT).

\section{Minimizing Total Flow Time}
\label{sec:flow}
The purpose of this section is to determine the optimal allocation of servers to jobs at every time, $t$, in order minimize the mean flow time of a set of $M$ jobs.  This is equivalent to minimizing the \emph{total} flow time for a set of $M$ jobs, and thus we consider minimizing total flow time for the remainder of this section.  We derive a closed-form for the optimal allocation function $\bm{\theta}^*(t)$ which defines the allocation for each job at any moment in time, $t$, that minimizes total flow time.  

To derive $\bm{\theta}^*(t)$, it helps to first narrow the search space of potential allocation policies by proving some properties of the optimal policy.
Specifically, knowing the order in which jobs are completed by the optimal policy greatly constrains the form of the policy and makes the problem of finding a closed-form much more tractable.
It is tempting to assume that the optimal policy completes jobs in Shortest-Job-First (SJF) order.  This intuition is based on the case where the ``resource'' consists of only a \emph{single} server.
In the single server case, it is known that the optimal policy completes jobs in SJF order.
The proof that SJF is the optimal completion order relies on a simple interchange argument:
Consider any time $t$ when an allocation policy does not allocate the server solely to the smallest job.  
By instead allocating the whole server to the smallest job at time $t$, the total flow time across jobs is reduced.

Unfortunately, in the case with many servers and parallelizable jobs, it is inefficient and typically undesirable to give the \emph{entire} system to a single job.
Hence, the usual simple interchange argument fails to generalize to the setting considered in this paper.
In general, it is not obvious what fraction of the system resources should be given to each job in the system, and it is unclear how this fraction should change as jobs depart the system over time.

Instead of an interchange argument, we present an alternative proof that the optimal completion order in a many server system with parallelizable jobs is SJF.
Our proof will require a sequence of theorems.
First, given any completion order, $C$, we will consider the policy $P$ which is optimal with respect to $C$.
We will prove several properties of this policy $P$, including finding an expression for the total flow time under $P$.
This will allow us to then optimize over the space of potential completion orders and conclude that the optimal completion order is SJF.

Once we know that the optimal completion order is SJF, we can derive an exact form of the optimal allocation function, $\bm{\theta}^*(t)$.  The theorems required in our analysis are outlined in Section \ref{sec:overview}.

\subsection{Overview of Our Results}
\label{sec:overview}
We begin by showing that the optimal allocation does not change between job departures, and hence it will suffice to consider the value of the allocation function only at times just after a job departure occurs.  This is stated formally as Theorem \ref{thm:constant}.
\begin{restatable}{thm}{constant}
    \label{thm:constant}
    Consider any two times $t_1$ and $t_2$ where, WLOG, $t_1<t_2$.  Let $m^*(t)$ denote the number of jobs in the system at time $t$ under the optimal policy.  If $m^*(t_1)=m^*(t_2)$ then
    $$\bm{\theta}^*(t_1) = \bm{\theta}^*(t_2).$$
\end{restatable}
\noindent Theorem \ref{thm:constant} is proven in Section \ref{sec:prelim}.

Another key property of the optimal allocation, which we refer to as the \emph{scale-free property}, states that for any job, $i$, job $i$'s allocation relative to jobs completed after job $i$ remains constant throughout job $i$'s lifetime.  It turns out that our scale-free property holds for an even more general class of policies.  This generalization is stated formally in Theorem \ref{thm:sf}. 
\begin{restatable}[Scale-free Property]{thm}{sf}
    \label{thm:sf}
    Consider any completion order, $C=(c_1, c_2, \ldots, c_M)$.  Let $\bm{\theta}^P(t)$ denote the allocation function of a policy $P$ which is optimal with respect to $C$.
    Let $t$ be a time when there are exactly $i$ jobs in the system and hence $m(t)=i$.  Consider any $t'$ such that $t'<t$.  Then,
    $$\frac{\theta_{c_i}^P(t')}{\sum_{j=1}^{i} \theta_{c_j}^P(t')} = \theta_{c_i}^P(t).$$
\end{restatable}
\noindent The scale-free property is important because it allows us to derive an expression for the total flow time under any policy $P$ where $P$ is optimal with respect to a given completion order $C$ (see Lemma \ref{lem:transform}).
Theorem \ref{thm:sf} is proven in Section \ref{sec:sf}.

Our next step is to minimize the expression from Lemma \ref{lem:transform} over the space of all completion orders to prove that the optimal completion order, $C^*$, is the SJF order.
This is stated in Theorem \ref{thm:sjf}.

\begin{restatable}[Optimal Completion Order]{thm}{sjf}
    \label{thm:sjf}
    The optimal policy follows the completion order, $C^*$, where 
    $$C^* = (1,2,3, \ldots, M)$$
    and hence jobs are completed in the Shortest-Job-First (SJF) order.
\end{restatable}
\noindent Since jobs are completed in SJF order, we can conclude that, at time $t$, the jobs left in the system are specifically jobs $1, 2, \ldots, m(t)$.
Theorem \ref{thm:sjf} is proven in Section \ref{sec:comp}.

We next derive the optimal allocation policy with respect to the optimal completion order $C^*$.
The first step to deriving the optimal allocation policy is to prove a size-invariant property which says that the optimal allocation to jobs $1, 2, \ldots, m(t)$ at time $t$ depends only on $m(t)$, not on the specific remaining sizes of the these jobs.
This is a counter-intuitive result because one might imagine that the optimal allocation should be different if two very differently sized jobs are in the system instead of two equally sized jobs.
The size-invariant property is stated in Theorem \ref{thm:indep}.
\begin{restatable}[Size-invariant Property]{thm}{indep}
    \label{thm:indep}
    Consider any time $t$.  Imagine two sets of jobs, $A$ and $B$, each of size $m$.  
    If $\bm{\theta}^*_A(t)$ and $\bm{\theta}^*_B(t)$ are the optimal allocations to the jobs in sets $A$ and $B$, respectively, we have that
    $$\bm{\theta}^*_A(t) = \bm{\theta}^*_B(t).$$
\end{restatable}

Theorem \ref{thm:indep} simplifies the computation of the optimal allocation, since it allows us to ignore the actual remaining sizes of the jobs in the system.  
We need only to derive one optimal allocation for each possible value of $m(t)$.  We prove Theorem \ref{thm:indep} in Section \ref{sec:opt}.

The consequence of the above results is that we can explicitly compute the optimal allocation function.
We are thus finally ready to state Theorem \ref{thm:opt}, which provides the allocation function for the optimal allocation policy which minimizes total flow time.
\begin{restatable}[Optimal Allocation Function]{thm}{opt}
    \label{thm:opt}
    At time $t$, when $m(t)$ jobs remain in the system,
    $$\theta^*_i(t) = \left(\frac{i}{m(t)}\right)^{\frac{1}{1-p}} - \left(\frac{i-1}{m(t)}\right)^{\frac{1}{1-p}} \qquad \forall 1\leq i \leq m(t).$$
\end{restatable}
Theorem \ref{thm:opt} is proven in Section \ref{sec:opt}.
Given the optimal allocation function $\bm{\theta}^*(t)$, we can also explicitly compute the optimal total flow time for any set of $M$ jobs.
This is stated in Theorem \ref{thm:optt}.
\begin{restatable}[Optimal Total Flow Time]{thm}{optt}
    \label{thm:optt}
    Given a set of $M$ jobs of size $x_1 > x_2 > \ldots > x_M$, the total flow time, $T^*$, under the optimal allocation policy $\bm{\theta}^*(t)$ is given by
    $$T^*=\frac{1}{s(N)}\sum_{k=1}^M  x_{k}\cdot\bigl[k  s(1 + \omega^*_{k})  - (k-1) s(\omega^*_{k})\bigr]$$
    where
    $$\omega_{k}^* = \frac{1}{\left(\frac{k}{k-1}\right)^{\frac{1}{1-p}} - 1} \qquad \forall 1 < k \leq M$$
    and
    $$\omega_1^* = 0$$
\end{restatable}

Note that the optimal allocation policy biases towards short jobs, but does not give strict priority to these jobs in order to maintain the overall efficiency of the system.  That is,
$$0<\bm{\theta}^*_1(t) <\bm{\theta}^*_2(t) < \ldots < \bm{\theta}^*_{m(t)}(t).$$
We thus refer to the optimal policy derived in Theorem \ref{thm:opt} as \emph{High Efficiency Shortest-Remaining-Processing-Time} or \emph{heSRPT}.

\subsection{A Property of the Optimal Policy}\label{sec:prelim}
In determining the optimal completion order of jobs, we first show that the optimal allocation function remains constant between job departures.
This allows us to think of the optimal allocation function as being composed of $M$ decision points where new allocations must be determined.
\constant*
\begin{proof}
    Consider any time interval $[t_1, t_2]$ during which no job departs the system, and hence $m^*(t_1) = m^*(t_2)$.  
    Assume for contradiction that the optimal policy is unique, and that $\bm{\theta}^*(t_1) \neq \bm{\theta}^*(t_2)$.  
    We will show that the mean flow time under this policy can be improved by using a \emph{constant} allocation during the time interval $[t_1,t_2]$, where the constant allocation is equal to the average value of $\bm{\theta}^*(t)$ during the interval $[t_1, t_2]$.

Specifically, consider the allocation function $\overline{\bm{\theta}^*(t)}$ where
$$\overline{\bm{\theta}^*(t)}=
    \begin{cases}
        \frac{1}{t_2 -t_1}\int_{t_1}^{t_2}\bm{\theta}^*(T)dT  &\forall t_1 \leq t \leq t_2\\
    \bm{\theta}^*(t)  &\mbox{otherwise}.
\end{cases}$$
Note that $\overline{\bm{\theta}^*(t)}$ is constant during the interval $[t_1,t_2]$.
    Furthermore, because $\sum_{i=1}^{m(t)}\theta_i^*(t) \leq 1$ for any time $t$, $\sum_{i=1}^{m(t)}\overline{\theta_i^*(t)} \leq 1$ at every time $t$ as well, and is therefore a feasible allocation function.
Because the speedup function, $s$, is a concave function, using the allocation function $\overline{\bm{\theta}^*(t)}$ on the time interval $[t_1,t_2]$ provides a greater (or equal) average speedup to every active job during this interval.  
Hence, the residual size of each job under the allocation function $\overline{\bm{\theta}^*(t)}$ at time $t_2$ is at most the residual size of that job under $\bm{\theta}^*(t)$ at time $t_2$.
There thus exists a policy which achieves an equal or lower total flow time than the unique optimal policy, but changes allocations only at departure times.
This is a contradiction.
\end{proof}
For the rest of the paper, we will therefore only consider allocation functions which change only at departure times.
The result of Theorem \ref{thm:constant} generalizes to the case where some policies which are optimal with respect to any completion order, $C$.
The same argument holds in this case because we can always improve a policy by ensuring that it changes only at departure times.

\subsection{The Scale-Free Property}
\label{sec:sf}
Consider any given completion order, $C$, and the policy $P$ which is optimal with respect to $C$.
Our goal is to characterize $P$ strongly enough that we can optimize over the space of all completion orders, $C$, and determine the optimal completion order $C^*$.
Hence, we now prove an interesting invariant of any policy $P$, which we call the scale-free property.
We will first need a preliminary lemma.
\begin{restatable}{lem}{scaling}
    \label{lem:scaling}
    Consider an allocation function $\bm{\theta}(t)$ which, at all times $t$ leaves $\beta$ fraction of the system unused.  That is,
    $$\sum_{i=1}^{m(t)} \theta_i(t) = 1-\beta \qquad \forall t.$$
    The total flow time under $\bm{\theta}(t)$ is equivalent to the total flow time under an allocation function $\bm{\theta}'(t)$ where
    $$\bm{\theta}'(t) = \frac{\bm{\theta}(t)}{1-\beta} \qquad \forall t$$
    in a system that runs at $\left(1-\beta\right)^p$ times the speed of the original system (which runs at rate 1).
\end{restatable}
\begin{proof}
    Is straightforward, see Appendix \ref{sec:morescaling}.
\end{proof}
\noindent Using Lemma \ref{lem:scaling} we can characterize the policy $P$ which is optimal with respect to any completion order $C$.  Theorem \ref{thm:sf} states that a job's allocation relative to the jobs completed after it will remain constant for the job's entire lifetime.
\sf*
\begin{proof}
    We will prove this statement by induction on the overall number of jobs, $M$. 
    First, note that the statement is trivially true when $M=1$. 
    It remains to show that if the theorem holds for $M=k$, then it also holds for $M=k+1$.  
    
    Let $M=k+1$ and let $T^P_{i}$ denote the finishing time of job $i$ under the policy $P$.
    Recall that $P$ finishes jobs according to the completion order $C$, so $T^P_{c_{i+1}} < T^P_{c_i}$.  Consider a system which optimally processes $k$ jobs, which WLOG are jobs $c_1, c_2, \ldots, c_{M-1}$.  We will now ask this system to process an additional job, job $c_{M}$.
    From the perspective of the original $k$ jobs, there will be some constant portion of the system, $\theta^P_{c_M}$, used to process job $c_M$ on the time interval $[0,T^P_{c_M}]$.  The remaining $1-\theta^P_{c_M}$ fraction of the system will be available during this time period.  Just after time $T^P_{c_M}$, there will be $k$ jobs in the system, and hence by the inductive hypothesis the optimal policy will obey the scale-free property on the interval $(T^P_{c_{M-1}},T^P_{c_{1}}]$.
    
    Consider the problem of minimizing the total flow time of the $M$ jobs \emph{given any fixed value} of $\theta^P_{c_M}$ such that the completion order is obeyed.  We can write the total flow time of the $M$ jobs, $T^P$, as
    $$T^P=T^P_{c_M} + \sum_{j=1}^{M-1} T^P_{c_j}$$
    where $T^P_{c_M}$ is a constant.  Clearly, optimizing total flow time in this case is equivalent to optimizing the total flow time for $M=k$ jobs with the added constraint that $\theta^P_{c_M}$ is unavailable (and hence ``unused'' from the perspective of jobs $c_{M-1}$ through $c_{1}$) during the interval $[0, T^P_{c_M}]$.  By Lemma \ref{lem:scaling}, this is equivalent to having a system that runs at a fraction $\left(1-\theta^P_{c_M}\right)^p$ of the speed of a normal system during the interval $[0, T^P_{c_M}]$.

    Thus, for some $\delta>1$, we will consider the problem of optimizing total flow time for a set of $k$ jobs in a system that runs at a speed $\frac{1}{\delta}$ times as fast during the interval $[0,T^P_{c_M}]$.

    Let $T^P[x_{c_{M-1}}, \ldots, x_{c_{1}}]$ be the total flow time under policy $P$ of $k$ jobs of size $x_{c_{M-1}} \ldots x_{c_1}$.  Let $T^S[x_{c_{M-1}}, \ldots, x_{c_{1}}]$ be the total flow time of these jobs in a \emph{slow system} which always runs $\frac{1}{\delta}$ times as fast as a normal system.

If we let $T^S_i$ be the finishing time of job $i$ in the slow system, it is easy to see that 
$$T^P_i = \frac{T^S_i}{\delta}$$
since we can just factor out a $\delta$ from the expression for the completion time of every job in the slow system.  
    Furthermore, we see that
    $$T^P[x_{c_{M-1}}, \ldots, x_{c_{1}}] = \frac{T^S[x_{c_{M-1}}, \ldots, x_{c_{1}}]}{\delta}$$
    by the same reasoning.  Clearly, then, the allocation function which is optimal with respect to $C$ in the slow system, $\bm{\theta}^S(t)$, is equal to $\bm{\theta}^P(t)$ at the respective departure times of each job.  That is,
    $$\bm{\theta}^P(T^P_j) = \bm{\theta}^S(T^S_j) \qquad \forall 1 \leq j \leq M-1.$$
   
    We will now consider a \emph{mixed} system which is ``slow'' for some interval $[0,T^P_{c_M}]$ that ends before $T^S_{c_{M-1}}$, and then runs at normal speed after time $T^P_{c_M}$.
    Let $T^Z[x_{c_{M-1}}, \ldots, x_{c_{1}}]$ denote the total flow time in this mixed system and let $\bm{\theta}^Z(t)$ denote the allocation function which is optimal with respect to $C$ in the mixed system.
    We can write
    \begin{align*}
        T^Z[x_{c_{M-1}}, \ldots, x_{c_{1}}] &= k\cdot T^P_{c_M}\\
        &+ T^P[x_{c_{M-1}}-\frac{s(\theta_{c_{M-1}}^Z)T^P_{c_M}}{\delta}, \ldots x_{c_{1}} - \frac{s(\theta_{c_{1}}^Z)T^P_{c_M}}{\delta}].
    \end{align*}
Similarly we can write

    \begin{align*}T^S[x_{c_{M-1}}, \ldots, x_{c_{1}}] &= k\cdot T^P_{c_M}\\
        &+ T^S[x_{c_{M-1}}-\frac{s(\theta_{c_{M-1}}^S)T^P_{c_M}}{\delta}, \ldots x_{c_{1}} - \frac{s(\theta_{c_{1}}^S)T^P_{c_M}}{\delta}].
    \end{align*}

    Let $T^Z_{i}$ be the finishing time of job $i$ in the mixed system under $\bm{\theta}^Z(t)$.
    Since $k\cdot T^P_{c_M}$ is a constant not dependent on the allocation function, we can see that the optimal allocation function in the mixed system will make the same allocation decisions as the optimal allocation function in the slow system at the corresponding departure times in each system.  That is,
    $$\bm{\theta}^P(T^P_j) = \bm{\theta}^S(T^S_j) = \bm{\theta}^Z(T^Z_j) \qquad \forall 1 \leq j \leq M-1.$$
    
    By the inductive hypothesis, the optimal allocation function in the slow system obeys the scale-free property. Hence, $\bm{\theta}^Z(t)$ also obeys the scale-free property for this set of $k$ jobs given any fixed value of $\theta^P_{c_M}$.

    Let $T^P$ denote the total flow time in a system with $M=k+1$ jobs that runs at normal speed, where a $\theta^P_{c_M}$ fraction of the system is used to process job $c_M$.  For any value of $\theta^P_{c_M}$, we can write
    \begin{align*}
        T^P[x_{c_M}, x_{c_{M-1}}, \ldots, x_{c_{1}}] =&\\
        (k+1)\cdot T^P_{c_M} + T^P[x_{c_{M-1}}-&\frac{s(\theta_{c_{M-1}}^P)T^P_{c_M}}{\delta}, \ldots x_{c_{1}} - \frac{s(\theta_{c_{1}}^P)T^P_{c_M}}{\delta}].
    \end{align*}
    for some $\delta>1$.  Clearly, this expression would be minimized by $\bm{\theta}^Z(t)$ except for two things: the system runs at normal speed, and a $\theta^P_{c_M}$ fraction of the system is unavailable during the interval $[0,T^P_{c_M}]$.  By Lemma \ref{lem:scaling}, the mixed system and the system with unavailable servers are equivalent, and hence we can simply re-normalize all allocations from $\bm{\theta}^Z(t)$ by $1-\theta^P_{c_M}$ during this interval to account for this difference.  This yields the optimal allocation function in a normal speed system for a given value of $\theta^P_{c_M}$ on the interval $[0,T^P_{c_M}]$.  Since $\bm{\theta}^Z(t)$ obeys the scale-free property, the normalized allocation function clearly also obeys the scale-free property.  Hence, the optimal allocation function for processing the $M=k+1$ jobs obeys the scale-free property.  This completes the proof by induction.
\end{proof}
\begin{restatable}{definition}{sfimpact}
    \label{rem:sf}
    \noindent The scale-free property tells us that for any completion order $C$, under the policy $P$ which is optimal with respect to $C$, a job's allocation relative to the jobs completed after it is constant.  Hence, for any job, $c_i$, there exists a {\bf scale-free constant} $\omega^P_i$ where, for any $t<T^P_{c_i}$
$$\omega^P_{i} = \frac{\sum_{j=1}^{i-1} \theta^P_{c_j}(t)}{\theta^P_{c_i}(t)}.$$
    Note that we define $\omega_1 =0$.  Let $\bm{\omega}^P=(\omega^P_1, \omega^P_2, \ldots, \omega^P_M)$ denote the scale-free constants corresponding to each job.
\end{restatable}

\subsection{Finding the Optimal Completion Order, $C^*$}
\label{sec:comp}
We will now make use of the scale-free property to find the optimal completion order, $C^*$.
We again consider any given completion order, $C$, and the policy $P$ which is optimal with respect to $C$.
In Lemma \ref{lem:transform} below, we derive an expression for the total flow time under the policy $P$ as a function of $\bm{\omega}^P$.
Finally, in Theorem \ref{thm:sjf}, we minimize this expression over all completion orders, $C$, to find the optimal completion order, $C^*$.

\begin{restatable}{lem}{transform}
    \label{lem:transform}
Consider a policy $P$ which is optimal with respect to the completion order $C$.  We define
    $$\omega^P_{i} = \frac{\sum_{j=1}^{i-1} \theta^P_{c_j}(t)}{\theta^P_{c_i}(t)} \qquad \forall 1<i\leq M, 0\leq t<T^P_{c_i}$$
    and $\omega^P_1 = 0$. We can then write the total flow time under policy $P$ as function of $\bm{\omega}^P = (\omega^P_1, \omega^P_2, \ldots, \omega^P_M) $ as follows
\begin{eqnarray*}
T^P(\bm{\omega}^P)
    &=& \frac{1}{s(N)}\sum_{k=1}^M  x_{c_k} \cdot \bigl[k  s(1 + \omega^P_{k})  - (k-1) s(\omega^P_{k})\bigr]\\
\end{eqnarray*}
\end{restatable}
\begin{proof}
To analyze the total flow time under $P$ we will relate $P$ to a simpler policy, $P'$, which is much easier to analyze.
We define $P'$ to be
    $$\theta^{P'}_i = \theta^{P'}_i(t) = \theta^P_i(0) \qquad \forall 1 \leq i \leq M.$$
Importantly, each job receives some initial optimal allocation at time 0 which does not change over time under $P'$.
Since allocations under $P'$ are constant we have that
\begin{eqnarray*}
T^{P'}_k  &=& \frac{x_k}{s(\theta^{P'}_k)} .
\end{eqnarray*}
We can now derive equations that relate $T^P_k$ to $T^{P'}_k$.

By Theorem \ref{thm:sf}, during the interval $[T^{P'}_{c_{k+1}}, T^{P'}_{c_{k}}]$,  
$$\omega^P_i = \omega^{P'}_i \qquad \forall 1 \leq i \leq k.$$
Note that a fraction of the system $\sum_{i=k+1}^M \theta^{P'}_{c_i}$ is unused during this interval, and hence by Lemma \ref{lem:scaling}, we have that
\begin{eqnarray*}
T^{P'}_{c_k} - T^{P'}_{c_{k+1}} &=& \frac{T^P_{c_k} - T^P_{c_{k+1}}}{s(\theta^{P'}_{c_1} + \cdots + \theta^{P'}_{c_k})} \\
\end{eqnarray*}
    Let $\alpha_k = \frac{1}{s(\theta^{P'}_{c_1} + \cdots + \theta^{P'}_{c_k})}$ denote the scaling factor during this interval.

    If we define $x_{c_{M+1}} = 0$ and $T^P_{c_{M+1}}=0$, we can express the total flow time under policy $P$, $T^P$, as
\begin{align*}
    T^P &= M T^P_{c_M} + (M-1) (T^P_{c_{M-1}} - T^P_{c_M}) + \cdots + 2 (T^P_{c_{2}} - T^P_{c_{3}})  + (T^P_{1} - T^P_{c_{2}})\\
    &= \sum_{k=1}^M  k (T^P_{c_k} - T^P_{c_{k+1}})  \\
    &= \sum_{k=1}^M  k \frac{T^{P'}_{c_k} - T^{P'}_{c_{k+1}}}{\alpha_k} .
\end{align*}
    We can now further expand this expression in terms of the job sizes, using the fact that $s(ab)=s(a)\cdot s(b)$, as follows:
\begin{align*}
    T^P =& \sum_{k=1}^M  k \frac{\frac{x_{c_k}}{s(\theta^{P'}_{c_{k}}N)} - \frac{x_{c_{k+1}}}{s(\theta^{P'}_{c_{k+1}}N)}}{\alpha_k} \\
    =& \frac{1}{s(N)}\sum_{k=1}^M  k\cdot \bigl[x_{c_k} s(1+\omega^{P'}_k) -x_{c_{k+1}} s(\omega^{P'}_{k+1})\bigr]  \\
    =&\frac{1}{s(N)}\sum_{k=1}^M  x_{c_k}\cdot \bigl[k  s(1 + \omega^P_{k})  - (k-1) s(\omega^P_{k})\bigr]
\end{align*}
as desired.
\end{proof}

We now have an expression for the total flow time of a policy $P$ which is optimal with respect to a given completion order $C$.  Next, we use this expression to derive the optimal completion order, $C^*$, in Theorem \ref{thm:sjf}.

\sjf*
\begin{proof}
    Consider a policy $P$ that is optimal with respect to any given completion order $C$.  
    Let $T^P(\bm{\omega}^P)$ be the expression for total flow time for $P$ from Lemma \ref{lem:transform}.
    Our goal is to find a closed-form expression for $\bm{\omega}^P$, and then minimize $T^P(\bm{\omega}^P)$ over the space of possible completion orders.
    A sufficient condition for finding $\bm{\omega}^P$ is that jobs are completed according to $C$ \emph{and} that the following first-order conditions are satisfied. 
\begin{eqnarray*}
\frac{\partial T^P}{\partial \omega^P_{k}} = 0
\end{eqnarray*} 

    Note that the second order conditions are satisfied trivially.
    These first order conditions are sufficient, but not \emph{necessary}, since $P$ may lie on a boundary of the space which is imposed by the completion order.
For now, we will ignore the constraints on $\bm{\omega}^P$ imposed by the completion order.
    That is, we will consider minimizing the function $T^P$ without constraining $\bm{\omega}^P$ to follow the completion order $C$.
We refer to this as the \emph{relaxed minimization problem}, and we call the solution to this optimization the \emph{relaxed minimum} of the function $T^P$.
    Let $\bm{\Theta}^{P}(t)$ denote the allocation function of the relaxed minimum of $T^P$.
    Crucially, note that the value of $T^P$ under the relaxed minimum is a \emph{lower bound} on the total flow time \footnote{This value may not correspond to a total flow time under $\bm{\Theta}^{P}(t)$. It is just a value of the function $T^P$.} under $P$.
    We will use the solution to the relaxed minimization problem to argue about the optimal completion order, $C^*$.

The first order conditions give
$$ks'(1+\omega^P_{k}) - (k-1) s'(\omega^P_{k}) = 0 \qquad \forall 1 \leq k \leq M$$
and hence
$$\omega_{k}^P = \frac{1}{\left(\frac{k}{k-1}\right)^{\frac{1}{1-p}} - 1} \qquad \forall 1 < k \leq M$$

    We can show that the values of $\omega^P_k$ and $\bm{\Theta}^{P}_{c_k}(t)$ are increasing in $k$ (see Appendix \ref{sec:addl}).  Furthermore, the coefficient of $x_{c_k}$ in $T^P(\bm{\omega}^P)$ is
    $$k  s(1 + \omega^P_{k})  - (k-1) s(\omega^P_{k}),$$
    which is increasing in $k$ (see Appendix \ref{sec:addl}).
    This implies that the completion order which produces the best relaxed minimum is SJF, since SJF matches the smallest job sizes with the largest coefficients in $T^P$.
    In addition, since $\bm{\Theta}^{P}_{c_k}(t)$ is increasing in $k$ for any completion order, under this allocation function smaller jobs always have larger allocations than larger jobs.
    This implies that the relaxed minimum for $T^P$ under SJF respects the SJF completion order.
    Because the solution to the relaxed minimization problem is feasible, it is also optimal for the constrained minimization problem.
    Furthermore, because the SJF completion order has the best relaxed minimum, the relaxed minimum for the SJF completion order must be the optimal allocation function.
    We thus conclude that the SJF completion order is the completion order of the optimal policy.
\end{proof}
\subsection{Finding the Optimal Allocation Function}
\label{sec:opt}
Now that we know the optimal completion order, $C^*$, we can compute the optimal allocation function $\bm{\theta}^*(t)$.  This computation begins with the interesting observation that the optimal allocation function does not directly depend on the sizes of the $M$ jobs.  This is stated in Theorem \ref{thm:indep}.
\indep*
\begin{proof}
Recall from the proof of Theorem \ref{thm:sjf} that the optimal allocation function satisfies the following first order conditions
$$ks'(1+\omega^P_{k}) - (k-1) s'(\omega^P_{k}) = 0 \qquad \forall 1 \leq k \leq M$$
and always completes jobs in SJF order.
    Note that these conditions do not explicitly depend on any job sizes.  Hence, while the value of the optimal allocation function may depend on \emph{how many} jobs are in the system, given any two sets of jobs, $A$ and $B$, which consist of $m$ jobs at time $t$, the optimal allocation function for set $A$ will be equal to the optimal allocation function for set $B$ at time $t$.
\end{proof}
\noindent We now use our knowledge of the optimal completion order to derive the optimal allocation function in Theorem \ref{thm:opt}.

\opt*
\begin{proof}
    We can now solve a system of equations to derive the optimal allocation function.
    Consider a time, $t$, when there are $m(t)$ jobs in the system.  Since the optimal completion order is SJF, we know that the jobs in the system are specifically jobs $1,2, \ldots, m(t)$.  We know that the allocation to jobs $m(t) + 1, \ldots, M$ is 0, since these jobs have been completed.  Hence, we have that
    $$\theta^*_1 + \theta^*_2 + \ldots + \theta^*_{m(t)} = 1$$
    Furthermore we have $m(t) - 1$ constraints provided by the expressions for $\omega^*_2, \omega^*_3, \ldots, \omega^*_{m(t)}$.
\begin{eqnarray*}
    \omega^*_2 =  \frac{\theta^*_1(t)}{\theta^*_2(t)} ,  \omega^*_3 = \frac{\theta^*_2(t) + \theta^*_1(t)}{\theta^*_3(t)},
    \cdots,\omega^*_{m(t)} = \frac{\theta^*_{m(t)-1}(t) + \cdots + \theta^*_1(t)}{\theta^*_{m(t)}(t)},
\end{eqnarray*}
These can be written as
\begin{eqnarray*}
    \theta^*_1(t) &=& \omega^*_2 \theta^*_2(t) \\
    \theta^*_1(t) + \theta^*_2(t) &=& \omega^*_3 \theta^*_3(t) \\
\cdots\\
    \theta^*_1(t) + \theta^*_2(t) + \cdots + \theta^*_{m(t)-1}(t)  &=& \omega^*_{m(t)} \theta_{m(t)} \\
    \theta^*_1 + \theta^*_2 + \ldots + \theta^*_{m(t)} &=& 1\\
\end{eqnarray*}

And then rearranged as
\begin{eqnarray*}
    \theta^*_{m(t)}(t) &=& \frac{1}{1 + \omega^*_{m(t)}} \\
    \theta^*_{m(t)-1}(t) &=& \frac{\theta^*_{m(t)}(t) \omega^*_{m(t)}}{1 + \omega^*_{m(t)-1}} = \frac {\omega^*_{m(t)}}{(1+\omega^*_{m(t)-1})(1+\omega^*_{m(t)})} \\
\cdots \\
    \theta^*_2(t) &=& \frac{\theta^*_3(t) \omega^*_3}{1 + \omega^*_{2}} = \frac {\omega^*_3 \cdots \omega^*_{m(t)}}{(1+\omega^*_{2}) \cdots (1+\omega^*_{m(t)})} \\
    \theta^*_1(t) &=& \frac{\theta^*_2(t) \omega^*_2}{1 + \omega^*_{1}} = \frac {\omega^*_{2} \cdots \omega^*_{m(t)}}{(1+\omega^*_{1})\cdots(1+\omega^*_{m(t)})} \\ 
\end{eqnarray*}
    We can now plug in the known values of $\omega^*_i$ and find that
    $$\bm{\theta}^{*}_{i}(t) = \left( \frac{i}{m(t)} \right)^{\frac{1}{1 - p}} - \left( \frac{i -1}{m(t)} \right)^{\frac{1}{1 - p}} \forall 1 \leq i \leq m(t)$$
    This argument holds for any $m(t)\geq 1$, and hence we have fully specified the optimal allocation function.
\end{proof}

Taken together, the results of this section yield an expression for total flow time under the optimal allocation policy.  This expression is stated in Theorem \ref{thm:optt}.
\optt*
Throughout the rest of the paper we will refer to the optimal allocation policy as heSRPT.

\section{Discussion and Evaluation}
\label{sec:eval}
We will now put the results of Section \ref{sec:flow} in context by examining the allocation decisions made by the optimal allocation policy, heSRPT, and comparing the performance of heSRPT to existing allocation policies from the literature.
In Section \ref{sec:balance}, we examine how heSRPT balances overall system efficiency with a desire to bias towards small jobs.
In Section \ref{sec:compare}, we next perform a numerical evaluation of heSRPT and several other allocation policies.
Finally, in Section \ref{sec:limits}, we will discuss some limitations of our results.
\subsection{What is heSRPT Really Doing?}
\label{sec:balance}
To gain some intuition about the allocation decisions made by heSRPT, consider what our system would look like if the speedup parameter $p$ was equal to $1$.
In this case, the many server system essentially behaves as one perfectly divisible server.  
We know that the optimal allocation policy in this case is SRPT which gives strict priority the shortest job in the system at every moment in time.
Any allocation which does not give all resources to the smallest job will increase mean flow time.

When $p<1$, however, there is now a trade-off.
In this case, devoting all the resources to a single job will decrease the total service rate of the system greatly, since a single job will make highly inefficient use of these resources.
We refer to the \emph{system efficiency} as the total service rate of the system, scaled by the number of servers.
We see that one can actually \emph{decrease} the mean flow time by allowing some sharing of resources between jobs in order to increase system efficiency.
The heSRPT policy correctly balances this trade-off between getting short jobs out of the system quickly, and maintaining high system efficiency.

In examining the heSRPT policy, it is surprising that while a job's allocation depends on the \emph{ordering} of job sizes, it does not depend on the \emph{specific} job sizes.  
One might assume that, given two jobs, the optimal allocation should depend on size difference between the small job and the large job.  
However, just as SRPT gives strict priority to small jobs without considering the specific sizes of the jobs, heSRPT's allocation function does not vary with the job size distribution. 
The intuition for why this is true is that, given that the optimal completion order is SJF (see Theorem \ref{thm:sjf}), changing the size of a job will not affect the trade-off between SRPT and system efficiency.
As long as the \emph{ordering} of jobs remains the same, the same allocation will properly balance this trade-off.
Although the SRPT policy is similarly insensitive to the specific sizes of jobs, it is very surprising that we get the same insensitivity given parallelizable jobs in a multi-server system.

\subsection{Numerical Evaluation}
\begin{figure*}[ht]
\centering
\includegraphics[width=.9\textwidth]{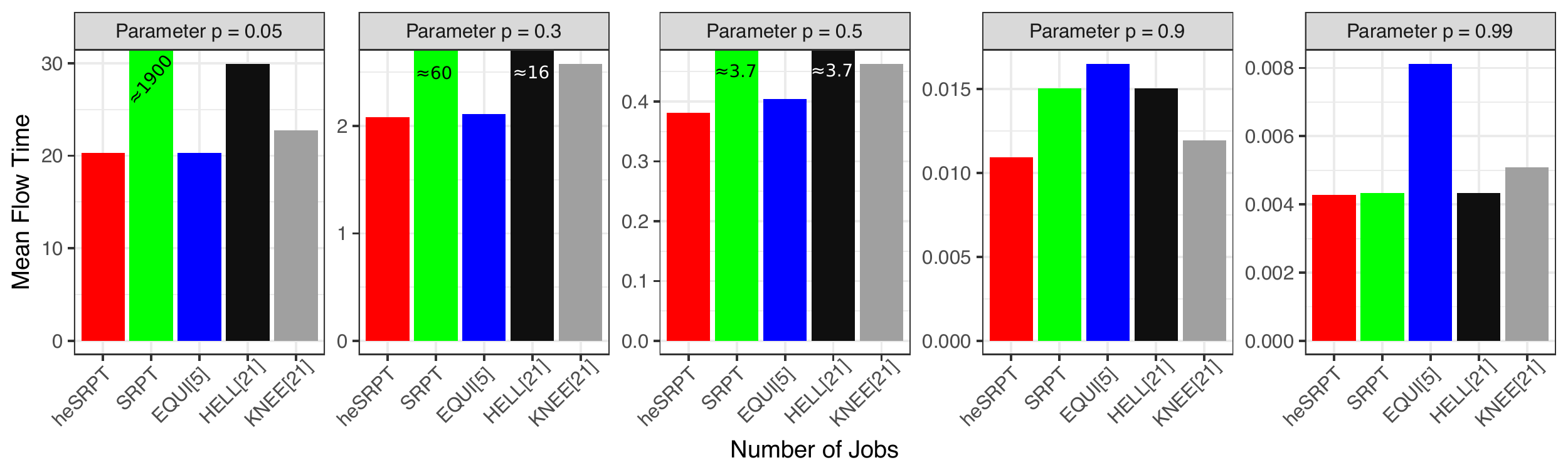}
    \caption{A comparison of heSRPT versus other allocation policies found in the literature.  Each policy is evaluated on a system of $N=1,000,000$ servers and a set of $M=500$ jobs whose sizes are drawn from a Pareto distribution with shape parameter $1.5$.  Each graph shows the mean flow time under each policy with various values of the speedup parameter, $p$.  The heSRPT policy outperforms every competitor by at least 30\% in at least one case.}
\label{fig:bars}
\end{figure*}

\label{sec:compare}
While heSRPT is the optimal allocation policy, it is still interesting to compare the mean flow time under this policy with the mean flow time under other policies from the literature.
Although we have a closed-form expression for the optimal total flow time under heSRPT (see Theorem \ref{thm:optt}), we wish to compare heSRPT to policies which do not necessarily emit closed-form expressions for total flow time.
Hence, we will perform a numerical analysis of these competitor policies on sets of jobs with randomly generated sizes, and for various levels of parallelizability (values of the speedup parameter $p$).

We compare heSRPT to the following list of competitor policies:
\begin{itemize}[leftmargin=.35cm]
    \item {\bf SRPT} allocates the entire system to the single job with shortest remaining processing time.  While this is known to be optimal when $p=1$ (see Section \ref{sec:balance}), we expect this policy to perform poorly when jobs make inefficient use of servers.
    \item {\bf EQUI} allocates an equal fraction of the system resources to each job at every moment in time.  This policy has been analyzed through the lens of competitive analysis \cite{edmonds2009scalably,Edmonds1999SchedulingIT} in similar models of parallelizable jobs, and was shown to be optimal in expectation when job sizes are unknown and exponentially distributed \cite{berg2018}.
        Other policies such as \emph{Intermediate-SRPT} \cite{im2016competitively} reduce to EQUI in our model where the number of jobs, $M$, is assumed to be less than the number of servers, $N$.
    \item {\bf HELL} is a heuristic policy proposed in \cite{lin2018model} which, similarly to heSRPT, tries to balance system efficiency with biasing towards short jobs.  HELL defines a job's efficiency to be the function $\frac{s(k)}{k}$.  HELL then iteratively allocates servers. In each iteration, HELL identifies the job which can achieve highest ratio of efficiency to remaining processing time, and allocates to this job the servers required to achieve this maximal ratio.  This process is repeated until all servers are allocated.  While HELL is consistent with the goal of heSRPT, the specific ratio that HELL uses is just a heuristic.
    \item {\bf KNEE} is the other heuristic policy proposed in \cite{lin2018model}.  
    KNEE considers the number of servers each job would require before its marginal reduction in run-time from receiving an additional server would fall below some threshold, $\alpha$.  
    The number of servers a job requires in order to reach this threshold is called a job's \emph{knee allocation}.
    KNEE then iteratively allocates servers.
    In each iteration, KNEE identifies the job with the lowest knee allocation and gives this job its knee allocation.
    This process is repeated until all servers are allocated.
    Because there is no principled way to choose this $\alpha$, we perform a brute-force search of the parameter space and present the results given the best $\alpha$ parameter we found.
    Hence, results for KNEE should be viewed as an optimistic prediction of the KNEE policy's performance.
\end{itemize}
We evaluate heSRPT and the competitor policies in a system of $N=1,000,000$ servers with a set of $M=500$ jobs whose sizes are drawn from a Pareto distribution with shape parameter $1.5$.  We set the speedup parameter $p$ to be 0.05, 0.3, 0.5, 0.9, or 0.99.  Each experiment is run with 10 different sets of randomly generated job sizes and we present the median of the mean flow times measured for each case.  The results of this analysis are shown in Figure \ref{fig:bars}.

We see that the optimal policy, heSRPT, outperforms every competitor allocation policy in every case as expected.
When $p$ is low, EQUI is very close to heSRPT, but EQUI is almost a factor of 2 worse than optimal when $p=0.99$.
Inversely, when $p=0.99$, SRPT is nearly optimal.
However, SRPT is an order of magnitude worse than heSRPT when $p=0.05$.
While HELL performs similarly to SRPT in most cases, it is only $50\%$ worse than optimal when $p=0.05$.
The KNEE policy is by far the best of the competitor policies that we consider.
In the worst case, when $p=0.3$, KNEE is roughly $30\%$ worse than heSRPT.
Note, however, that these results for KNEE require brute-force tuning of the allocation policy, and are thus optimistic about KNEE's performance.

\subsection{Limitations}
\label{sec:limits}
This paper derives the first closed-form of the optimal policy for allocating servers to parallelizable jobs of known size, yet we also acknowledge some limitations of our model.  
First, we note that in practice not all jobs will follow a single speedup function.  
Due to differences between applications or between input parameters, it will often be the case that one job is more parallelizable than another.
The complexity of allocating resources to jobs with multiple speedup functions has been noted in the literature \cite{berg2018, edmonds2009scalably}, and even when jobs follow a single speedup function the problem is clearly non-trivial.  
Next, we note that in our model all jobs are assumed to be present at time 0.  
While we believe that heSRPT might provide a good heuristic policy for processing a stream of arriving jobs, finding the optimal policy in the case of arrivals remains an open question for future work.  
Finally, we constrain the speedup functions in our model to be of the form $s(k) = k^p$, instead of accommodating more general functions.  
While Theorem \ref{thm:constant} can be shown to hold for any concave speedup function, it is unclear whether the other results of Section \ref{sec:flow} hold for arbitrary or even concave speedup functions.
Although our model makes the above assumptions, we believe that the general implication of our results holds --- the optimal allocation policy should bias towards short jobs while maintaining overall system efficiency.

\section{Conclusion}
Modern data centers largely rely on users to decide how many servers to use to run their jobs.
When jobs are parallelizable, but follow a sublinear speedup function, allowing users to make allocation decisions can lead to a highly inefficient use of resources.
We propose to instead have the system control server allocations, and we derive the first optimal allocation policy which minimizes mean flow time for a set of $M$ parallelizable jobs.
Our optimal allocation policy, heSRPT, leads to significant improvement over existing allocation policies suggested in the literature.
The key to heSRPT is that it finds the correct balance between overall system efficiency and favoring short jobs.
We derive an expression for the optimal allocation, at each moment in time, in closed-form.

\clearpage
\appendix
\section{Proof of Theorem \ref{thm:helrpt}}
\label{sec:moremake}
\helrpt*
\begin{proof}
    We wish to find $\bm{\gamma}^*(t)$ such that, for any jobs $i$ and $j$,
    $$\frac{x_i}{s(\gamma_i^*(t))} = \frac{x_j}{s(\gamma_j^*(t))} \qquad \forall t\leq T^*_{max}.$$
    This implies that for any $i$ and $j$
    $$\frac{x_i^{1/p}}{x_j^{1/p}} = \frac{\gamma^*_i(t)}{\gamma^*_j(t)}$$
    Furthermore, we know that
    $$\gamma_1^*(t) + \gamma_2^*(t) + \ldots + \gamma_M^*(t) = 1$$
    and thus, for any job $i$,
    $$\frac{x_i^{1/p} \gamma_i^*(t)}{x_i^{1/p}}+ \frac{x_2^{1/p} \gamma_i^*(t)}{x_i^{1/p}}+\ldots+ \frac{x_M^{1/p} \gamma_i^*(t)}{x_i^{1/p}}= 1$$
    Rearranging, we have
    $$\gamma^*_i(t) = \frac{x_i^{1/p}}{\sum_{j=1}^M x_j^{1/p}}$$
    as desired.
    Under this policy, the $T^*_{max}$ is equal to the completion time of any job, $i$, which is
    $$\frac{x_i}{s\biggl(\frac{x_i^{1/p}}{\sum_{j=1}^{M}x_j^{1/p}}\biggr)}= ||\bm{X}||_{1/p}$$
    as desired.

\end{proof}
\section{Proof of Lemma \ref{lem:scaling}}
\label{sec:morescaling}
\scaling*
\begin{proof}
    Consider the policy $\bm{\theta}'(t)$ in a system which runs at $(1-\beta)^p = s(1-\beta)$ times the speed of the original system.
    The service rate of any job in this system is
    $$ s(1-\beta) \cdot s(\bm{\theta}'(t) \cdot N) = s((1-\beta) \cdot \bm{\theta}'(t) \cdot N) = s(\bm{\theta}(t) \cdot N).$$
    Since the service rate of any job at any time $t$ is the same under $\bm{\theta}(t)$ as it is under $\bm{\theta}'(t)$ in a system which is $(1-\beta)^p$ times as fast, the systems are equivalent and will have the same mean flow time.
\end{proof}
\section{Additional Lemmas}
\label{sec:addl}
\begin{restatable}{lem}{om}
    \label{lem:om}
We define
    $$\omega_{k}^P = \frac{1}{\left(\frac{k}{k-1}\right)^{\frac{1}{1-p}} - 1} \qquad \forall 1 < k \leq M$$
    and $\omega^P_1 = 0$. Then $\omega^P_k$ is increasing for $1 \leq k \leq M$. 
\end{restatable}
\begin{proof}
    Note that since $\frac{k}{k-1}$ is decreasing in $k$,
    the denominator of the expression for $\omega^P_k$ is decreasing for $1 < k \leq M$.
    Hence $\omega^P_k$ is increasing when $1 < k \leq M$, and since $\omega^P_1 = 0 < \omega^P_2$, $\omega^P_k$ is increasing for $1 \leq k \leq M$.
\end{proof}

\begin{restatable}{lem}{theta}
    Consider any completion order $C$ and let $P$ be the allocation policy which is optimal with respect to $C$.  Let $\Theta^P(t)$ be the allocation function of the relaxed minimum of $T^P$.  Then for any $t$, $\Theta^P_{c_k}(t)$ is increasing in $k$ for $1 \leq k \leq m(t)$.
\end{restatable}
\begin{proof}
    Following the same argument as Theorem \ref{thm:opt}, we can see that the expression for $\Theta^P(t)$ is
    $$\Theta^P_i(t) = \left(\frac{i}{m(t)}\right)^{\frac{1}{1-p}} - \left(\frac{i-1}{m(t)}\right)^{\frac{1}{1-p}} \qquad \forall 1 \leq i \leq m(t).$$
    Note that $\frac{1}{1-p} > 1$, so $i^{\frac{1}{1-p}}$ is convex.  Hence, 
    $$i^{\frac{1}{1-p}} - (i-1)^{\frac{1}{1-p}}$$
    is increasing in $i$.  This implies that $\Theta^P_i(t)$ is increasing in $i$ for $1 \leq i \leq m(t)$.
\end{proof}

\begin{restatable}{lem}{coef}
We define
    $$\omega_{k}^P = \frac{1}{\left(\frac{k}{k-1}\right)^{\frac{1}{1-p}} - 1} \qquad \forall 1 < k \leq M.$$
    and $\omega_1^P = 0$.  Then the expression
    $$\Delta(k) = k  s(1 + \omega^P_{k})  - (k-1) s(\omega^P_{k})$$
    is increasing in $k$.
\end{restatable}
\begin{proof}
    Let $c=\frac{1}{1-p}$.  We can rewrite $\Delta(k)$ as
    $$(k^c - (k-1)^c)^{1-p}.$$
    We then have that
    $$\frac{d}{dk} \Delta(k) = (1-p)(k^c - (k-1)^c)^{-p}(c(k^{c-1}-(k-1)^{c-1})).$$
    This derivative is positive for all $k$, and hence $\Delta(k)$ is increasing in $k$.
\end{proof}

\bibliography{bibliography}{}
\bibliographystyle{plain}

\end{document}